\documentclass[oneside,a4paper,reqno]{amsart}
\usepackage{amssymb}
\newtheorem{lemma}{Lemma}
\newtheorem{theorem}{Theorem}

\theoremstyle{remark}
\newtheorem*{remark}{Remark}
\newcommand{\tr}{{\rm tr\, }}
\newcommand{\te}[2]{\parbox[b]{#1 cm}{ \centering #2}}
\DeclareMathOperator{\SL}{SL}
\DeclareMathOperator{\ESL}{ESL}
\DeclareMathOperator{\GL}{GL}
\begin{document}
\begin{titlepage}
\begin{center}
\LARGE Spectra of phase point operators in odd prime dimensions and the extended Clifford group
\end{center}
\vspace{1 cm}
\begin{center} D. M. Appleby\footnote{D.M.Appleby@qmul.ac.uk}\\ Department of Physics,
Queen Mary, University of London, Mile End Road, London E1 4NS, UK
 \end{center}
 \vspace{0.5 cm}
\begin{center}  Ingemar Bengtsson\footnote{ingemar@physto.se}\\Stockholm University, AlbaNova, Fysikum,
106 91 Stockholm, Sweden 
\end{center}
 \vspace{0.5 cm}
\begin{center} S. Chaturvedi\footnote{scsp@uohyd.ernet.in}\\
Institute for Mathematical Sciences, Imperial College London,\\ London SW7
2BW, UK,\\
and
 School of Physics, University of Hyderabad, Hyderabad 500 046,
India
 \end{center}
\vspace{4 cm}
\begin{center}
\parbox{12 cm }{\begin{center} \textbf{Abstract} \end{center}
\vspace{0.15 cm}
We analyse the role of the Extended Clifford group in classifying the spectra of phase point operators within the 
framework laid out by Gibbons et al for setting up Wigner distributions on discrete phase spaces based on finite 
fields. To do so we regard the set of all the discrete phase spaces as a symplectic vector space over the finite 
field. Auxiliary results include a derivation of the conjugacy classes of ${\rm ESL}(2, \mathbb{F}_N)$.}
\end{center}
\end{titlepage}
\newpage
\section{Introduction}
The Wigner distribution, introduced by Wigner in 1932 \cite{1} for the purpose of associating a phase space distribution to quantum systems with $\mathbb{R}^N$ as the configuration space, has, over the years, been extended and generalised in many directions
 \cite{2}-\cite{5}, differing 
%from each other 
in the aspects of the orginal definition that one chooses to retain. Among these efforts, those aimed at setting up Wigner distributions for finite state systems have, of late, been of great interest because of their potential relevance to quantum information theory
% \cite{6}, 
\cite{7}. For the case of quantum systems with Hilbert space dimension $N$ equal to power of a prime, Gibbons et al \cite{5}  developed a formulation for Wigner distributions which elegantly ties up Wigner distributions or equivalently the phase point operators with mutually unbiased bases    \cite{8}, \cite{9} available in such dimensions. This approach yields not one but many possible definitions of phase point operators and hence of Wigner distributions and Gibbons et al propose a scheme for classifying them on the basis of the action of the Clifford group.  In the present work we examine 
%the question of 
this classification, 
% of these definitions 
taking the eigenvalue spectra of the phase point operators as the sole criterion. We highlight the role which the extended Clifford group plays in this context. In other words, 
%much 
in line with the  observations in  \cite{10} in the context of SIC POVMS \cite{11}, we argue 
that, in so far as the spectra of the phase point operators is concerned, what is relevant is not the Clifford Group, but the extended Clifford group. This has the consequence that the number of 'similarity' classes in the sense of spectra is 
%typically much 
smaller than that given by the Clifford group action alone---for $N=7$ it is $210$ as opposed to $360$. 

A brief summary of this work is as follows: In Section II we 
%briefly 
recapitulate the ideas developed by Gibbons et al leading to the definition of generalised Wigner distributions. 
We restrict ourselves to $N$ an odd prime, and set up the notation and terminology used later. In Section III,  we develop a convenient scheme for labelling the phase point operators by elements of $V_{N+1}(\mathbb{Z}_N)$, an $N+1$ dimensional vector space over the finite field $\mathbb{Z}_N$, and transcribe the actions of the Heisenberg-Weyl (H-W) group and the Clifford group in terms of linear transformations on $V_{N+1}(\mathbb{Z}_N)$. By a suitable choice of basis in $V_{N+1}(\mathbb{Z}_N)$, the scheme developed here is shown to yield a convenient way of labelling the `affine planes', sets of phase point operators related to each other by the H-W group action, as elements of $V_{N-1}(\mathbb{Z}_N)$. This machinery is put to use in Section IV to study the action of the Clifford group on the set of phase point operators and affine planes with the purpose of classifying them into orbits under the Clifford group. In Section V, we examine the spectra of the phase point operators belonging to distinct orbits and find that several orbits are degenerate in the spectral sense, suggesting the presence of an extra symmetry. In Section VI we  
%argue 
show that the extra symmetry is related to the operation of complex conjugation, and that the number of distinct spectra of phase point operators can be quantitatively understood in terms of the  Extended Clifford Group - the Clifford group augmented by the operation of complex conjugation. This section contains our principal results. We conclude in Section VI with a summary 
%of our results and some remarks concerning further 
and an outlook. The appendix contains a derivation of the conjugacy classes of the group 
$\ESL(2,\mathbb{F}_N)$, where $\mathbb{F}_N$ is a finite field. 

%Though, in this work, we do not discuss the case of prime powers, as a hint as to what to expect 
%in such situations, we give, as an appendix, relevant details of a similar analysis for the case 
%of two qubits and show that the overall picture in this case is consistent with that found in the 
%odd prime case.  

\section{Generalised Wigner distributions} 

We begin by recalling the salient features of the framework laid out by Gibbons et al for setting up discrete phase spaces and defining Wigner distributions thereon though not necessarily in the same logical sequence. Their general framework applies to quantum systems of prime power dimensions, but we 
confine ourselves to the odd prime case.  
%only so as not to clutter the presentation with notions and notations not required later.
\newcounter{spacerule}
\begin{list}{\emph{\arabic{spacerule}.}}{\setlength{\leftmargin}{.5 in}\setlength{\rightmargin}{.5 in}\usecounter{spacerule}}
\item \emph{Hilbert Space}:
Consider the Hilbert space $\mathcal{H}_N$ associated with an $N$-level quantum system and introduce in it an orthonormal basis $\{|q\rangle \colon q\in\mathbb{Z}_N \}$, referred to as the coordinate or the standard basis, and a momentum basis $\{|p)\colon p\in\mathbb{Z}_N \}$ related to it by a finite Fourier transformation:
\begin{equation}
|p)=\frac{1}{\sqrt{N}}\sum_{i=1}^{N}~\omega^{pq}|i\rangle  ;~~ \omega = e^{2\pi i/N}.
\end{equation} 
\item{{\it Phase Space}:} With the $N$-level quantum system we associate a  phase space consisting of points $(q,p)$ with $q,p\in\mathbb{Z}_N$, for convenience arranged in the usual Cartesian fashion. The fact that $ \mathbb{Z}_N$ is a field has the consequence that the phase space can be decomposed into $N$ parallel lines, striations, in $N+1$ ways. For each striation, we can select the line through the origin as  its representative and call it a `ray'. 
%( In standard terminology, the term `ray' denotes a collection. In contrast here we use the term 
%`ray' to denote a representative of the collection, the line through the origin, reserving the 
%term `striation' for the collection of parallel lines)  
We thus have $N+1$ rays where each ray consists of points of the form $(sq,sp)$ with $q,p$ fixed, not both zero and $s$ taking all values in $\mathbb{Z}_N$. For obvious reasons, the two special rays $(s,0)$ and  $(0,s)$ are denoted as horizontal and  vertical rays respectively and the corresponding striations as horizontal and vertical striations. We arrange the striations with these two as the extremes:
\begin{equation}
{\rm vertical},\cdots \cdots\cdots \cdots,{\rm horizontal}.
\end{equation}
\item\emph{Transformations on Phase Space}:  Two kinds of linear operations on the phase space which map the phase space to itself naturally arise: 
\begin{itemize}
\item Translations, $T(q_0,p_0)$, $N^2$ in number, have the effect of moving a phase point by an amount $(q_0,p_0)$. These operations constitute an abelian group and map a line in one striation to another in the same striation. Thus, for instance, starting from a ray, except the horizontal, all other members of the corresponding striation can be reached through $T(1,0)$. For the horizontal ray, one can achieve this through $T(0,1)$.
\item `Symplectic Rotations'
       \begin{equation} \left( \begin{array}{c} q^\prime  \\ p^\prime \end{array} \right) = \left( \begin{array}{cc} \alpha & \beta \\ \gamma & \delta \end{array} \right) \left( \begin{array}{c} q  \\ p \end{array} \right) ,
\end{equation}
with $\alpha,\beta,\gamma,\delta \in \mathbb{Z}_N$ and $\alpha\delta-\beta\gamma=1~~({\rm mod}~ N)$. Such $2\times 2$ matrices $S$ constitute the group $\SL(2,\mathbb{Z}_N)$, containing $N(N^2-1)$ elements.  Its generators are
\begin{equation} 
g_1 = \left( \begin{array}{cc} 1 & 1 \\ 0 & 1 \end{array} \right), \hspace{12mm}
g_2 = \left( \begin{array}{cc} 0 & 1 \\ -1 & 0 \end{array} \right)
\ . 
\label{4}
\end{equation}
These transformations map rays onto rays, hence the figurative appellation `rotations'. 
'Symplectic' means that when applied to $(q_1,p_1)$,$(q_2,p_2)$ they leave the symplectic product $(q_2p_1-q_1p_2)$ unchanged.
\end{itemize}  
The semi-direct product $\SL(2,\mathbb{Z}_N)\ltimes T(q_0,p_0)$ 
%obtained by putting the two groups together in the form of a semi-direct product 
will be of interest here. 

\item\emph{Heisenberg-Weyl group}:
Returning to the Hilbert space we define the operators $X$ and $Z$ on $\mathcal{H}_N$ through their actions on the coordinate basis 
$|q\rangle  $ (momentum basis $|p)$) as :
\begin{eqnarray} 
X|q\rangle = |q+1\rangle && \hspace{5mm} (~~ X|p) = \omega^{-p}|p)~~),\\ 
Z|q\rangle = \omega^{q}|q\rangle   &&\hspace{5mm} (~~ Z|p) = |p+1)~~).
\end{eqnarray}
Consider the set of $N^2$ operators
\begin{equation} D(q,p) = \tau^{qp} X^qZ^p \ , \end{equation}
where $\tau=-e^{i\pi/N}$. These operators, called displacement operators, obey 
%have the following properties
\begin{eqnarray}
D(q,p)^\dagger&=&D(-q,-p),\\
D(q,p)D(q^\prime,p^\prime) &=& \tau^{(q^\prime p - qp^\prime)}D(q+q^\prime, p+p^\prime),\\
Tr[D(q,p)]&=&N\delta_{q,0}\delta_{p,0},\\
Tr[D^\dagger(q,p)D(q^\prime,p^\prime)]&=&N\delta_{q,q^\prime}\delta_{p,p^\prime}.
\end{eqnarray}
  
\noindent This set of unitary operators forms an orthonormal basis in the space 
of all complex operators on $\mathcal{H}_N$ and furnishes a projective representation of the group of translations on the phase space. In the standard basis, they can be represented by 
matrices as
\begin{equation} [D(q,p)]_{\ell k} = \tau^{qp}\omega^{kp}\delta_{\ell,k+q} \ . \end{equation}
In particular, 
\begin{equation} [D(0,1)]_{\ell k}=Z = \omega^{kp}\delta_{\ell,k}; \ [D(1,0)]_{\ell k}=X = \delta_{\ell,k+1} \ . \end{equation}
Further $\{e^{i\phi}D(q,p)\}$ constitute the elements of the Weyl-Heisenberg group 
(or the Pauli Group).
 
\item\emph{The Clifford group}:
\noindent The Clifford group consists of all unitary matrices 
$U$ such that 
\begin{equation} UD(q,p)U^\dagger = e^{i\theta_{q,p}}D(f(q,p), g(q,p)), \ . \end{equation}
for suitable phases $e^{\theta_{q,p}}$ and functions $f,g$ It is a semi-direct product of the Weyl-Heisenberg group (playing the role of translations) with the group 
of $N$ dimensional unitary matrices $U(F)$ labelled by elements
\begin{equation}
F=\left( \begin{array}{cc} \alpha & \beta \\ \gamma & \delta \end{array} \right);~~ {\rm det}F=1~~ {\rm mod}~ N,
\end{equation}
of $\SL(2,\mathbb{Z}_N)$. The unitaries $U(F)$, defined up to a phase, act on $D(q,p)$
as follows: 
\begin{equation} U(F)D(q,p)U^\dagger(F) = D(\alpha q+\beta p, \gamma q+\delta p) , 
\label{16}
\end{equation} 
and, up to a phase, are explicitly given by \cite{AF}
\begin{equation}
U(F)= \begin{cases} 
\frac{1}{\sqrt{N}}\sum_{j,k}\tau^{\beta^{-1}(\alpha k^2-2jk+\delta j^2)}|j\rangle \langle k| \qquad & \beta\neq 0,\\
\sum_{j}\tau^{\alpha \gamma j^2}|\alpha j\rangle \langle j|\qquad & \beta= 0.
\end{cases}
\label{16a}
\end{equation}
We shall denote the unitaries corresponding to $g_1$ and $g_2$ by  $V$ and $W$. They generate all the $U$'s.
\item\emph{Mutually Unbiased Bases (MUB)}:
From the properties of the displacement operators listed above it is evident that two such operators 
$D(q,p),D(q^\prime,p^\prime)$ commute with each other if and only if the `symplectic product'
$(q^\prime p-q p^\prime)$ of $(q,p)$ and $(q^\prime,p^\prime)$ vanishes. This in turn implies that the operators $D(q,p)$ with $(q,p)$ lying on a ray commute with each other and hence their simultaneous eigenvectors furnish an orthonormal basis. Moreover, the orthonormal bases corresponding to different rays can be shown to be mutually unbiased with respect to each other \cite{5}. Since there are $N+1$ rays we are naturally led to the full set of $N+1$ mutually unbiased bases. This association between rays and MUBs is a key ingredient in the construction of generalised Wigner distributions as we shall see later. As to explicit expressions for the MUBs there are many choices depending on the choice of phases for the vectors. 
%In the present work w
We find the following choice, with the MUB's arranged in a specific order, to be %the most 
convenient \cite{AF}:  
$$M_0, M_1, \cdots, M_N,$$
with
 \begin{equation} M_m  =\begin{cases} \{|m,r\rangle  \colon r=0,\cdots,N-1\}\qquad & m=0,\cdots,N-1\\
  \{|\infty,r\rangle  \colon r=0,\cdots,N-1\} \qquad & m=N
  \end{cases}  
\end{equation}
and
\begin{equation}
|m,r\rangle   = \begin{cases}
V^m|r\rangle  \qquad & m=0,\cdots,N-1\\  W|r\rangle  \qquad & m = \infty
\end{cases}
\end{equation}
where $V$ and $W$ denote unitaries corresponding to $g_1$ and $g_2$ respectively.
For the sets obtained by replacing each vector in the ordered sets  $ M_0, M_1, \cdots, M_N$ by its projector we will use the notation  
$$
\mathcal{M}_0, \mathcal{M}_1, \cdots, \mathcal{M}_N.$$
Explicit expressions for the projectors are given below \cite{AF}:
\begin{equation}
|m,r\rangle \langle m,r|=
\begin{cases}\frac{1}{N}\sum_{j}\omega^{-rj} D(mj,j)\qquad &m\neq \infty,
\\
\frac{1}{N}\sum_{j}\omega^{-rj}D(j,0) \qquad &m=\infty.
\end{cases}
\label{21}
\end{equation}The MUB's at the extremes can clearly be identified as those associated with vertical and horizontal rays respectively. Again, as with striations, the vectors within each MUB, except for the last one, are related to each other (up to phases) by the action of $D(1,0)$. Those in the last one are related to each other (up to phases) by $D(0,1)$: 
\begin{align}
D(1,0)\colon|m,r\rangle & \to  \begin{cases}|m,r+1\rangle \qquad & m\neq \infty
\\
|m,r\rangle \qquad & m =\infty,
\end{cases}
\label{22}
\\
D(0,1)\colon |m,r\rangle & \to
\begin{cases}|m,r-m\rangle,\qquad & m\neq \infty,
\\
 |m,r-1\rangle \qquad & m = \infty
\end{cases}
\label{23}
\end{align}
 
\item\emph{Quantum Nets}: We have seen that each ray has a MUB associated with it. We have also seen that the passage from one to another member of the striation by translation, has a complete correspondence with the passage from one vector in the MUB to another through appropriate displacement operators. We now assign a specific vector (more precisely the corresponding projector) in the MUB to each ray and assign other vectors in the MUB to the other parallel lines parallel to the ray using the correspondence mentioned above. This assignment of vectors to lines in the phase space consistent with translations leads to a Quantum Net. Clearly, as there are $N+1$ rays and for each ray $N$ vectors to choose from the corresponding MUB's, there are  $N^{N+1}$ different ways of assigning vectors to the lines in phase space and hence $N^{N+1}$ distinct quantum nets. Now given a quantum net we can generate another unitarily equivalent by applying one of the $N^2$ displacement operators to all the vectors. This then means that we can divide $N^{N+1}$ quantum nets into $N^{N-1}$ equivalence classes containing   $N^2$ nets each. 
\item\emph{Phase Point Operators}: Having assigned a vector to each line in the phase space we are now in a position to assign a Hermitian operator $\hat{A}(q,p)$ to each point $(q,p)$ in the phase space by adding up all the projectors $P_\lambda$ associated with the $N+1$ lines $\lambda$ passing through that point and subtracting the Identity operator from the sum to make its trace equal 1:
\begin{equation}
\widehat{A}(q,p) = \sum_{\lambda} P_\lambda - \mathbb{I}~~~~;~~~~ {\rm Tr}\widehat{A}(q,p)=1.
\end{equation}
%Using only the combinatorial properties of the lines and the fact that the sum of the projectors 
%in a striation equal $\mathbb{I}$ o
One can invert this relation to obtain the projectors associated with a line $\lambda$ in terms of the phase point operators \begin{equation}
P_\lambda = \frac{1}{N}\sum_{(q,p)\in \lambda}\widehat{A}(q,p).
\end{equation}
Further using the mutually unbiasedness of the projectors one obtains:
\begin{equation}
{\rm Tr}[\widehat{A}(q,p)\widehat{A}(q^\prime,p^\prime)] = N\delta_{q,q^\prime}\delta_{p,p^\prime}.
\end{equation}
From the above discussion it is clear that the totality $\{\widehat{A}\}$ of $N^{N+1}$ phase point operators are obtained by picking one vector from each of the MUBs in all possible ways.

\item\emph{Affine Planes}: Given the phase point operator at say $\widehat{A}(0,0)$ we can generate $N^2$ operators $\widehat{A}(q,p)$ by applying all the $N^2$ displacement operators to it. Such a collection is taken to define an affine plane.  We thus have the following correspondences between the structures at the phase space level and those at the level of the Hilbert space:
\begin{eqnarray}
{\rm Points}~~ (q,p)&\longleftrightarrow& {\rm Phase ~point ~operators}~~\widehat{A}(q,p),\nonumber\\
{\rm Lines}~~ \lambda &\longleftrightarrow& {\rm Projectors},\\
{\rm Phase~ Space}~~  &\longleftrightarrow& {\rm Affine~ Plane}.  \nonumber
\end{eqnarray}
%It is also evident that t
The full set of $N^{N-1}$ affine planes can thus be obtained by dividing $N^{N+1}$ $\widehat{A}$'s into subsets, related to each other by %the elements of 
the Heisenberg-Weyl group.

\item\emph{Wigner Distributions}:
 With the phase point operators at hand one can define Wigner distribution corresponding to a given density operator $\rho$ as 
\begin{equation}
W_{\rho}(q,p)={\rm Tr}[\rho\widehat{A}(q,p)].
\end{equation}   
which, among others, have the desired marginals property: 
\begin{equation}
 \frac{1}{N}\sum_{(q,p)\in \lambda}W_{\rho}(q,p)={\rm Tr}[\rho P_\lambda].
\end{equation} 
i.e. a Wigner distribution averaged over a line gives a probability.
\end{list}
\noindent
\section{A convenient labelling scheme for phase point operators and affine planes}
As discussed earlier, a selection of projectors, one from each $ \mathcal{M}_0, \mathcal{M}_1, \cdots, \mathcal{M}_N$, defines a phase point operator $\widehat{A}$. We can therefore label a phase point operator uniquely 
by an array $(r_0,r_1,\cdots,r_N)$ :
$$\widehat{A}\leftrightarrow (r_0,r_1,\cdots,r_N).$$
This permits us to think of each phase point operator as a vector in $V_{N+1}(\mathbb{Z}_N)$, an $N+1$ dimensional vector space over $Z_N$. We now examine how they transform under the Clifford group. To this end we look at the action of the generators   $D(1,0)$ and $D(0,1)$ of the Heisenberg-Weyl group 
% and those of the generators $V$ and $W$ 
as well as the action of the symplectic group.

From $(\ref{22})$ and $(\ref{23})$ it follows that : 
\begin{eqnarray} 
D(1,0) &:& (r_0,r_1,\cdots,r_N) \rightarrow (r_0,r_1,\cdots,r_N) + (1,1,\cdots,1,0), \\
 D(0,1) &:&(r_0,r_1,\cdots,r_N) \rightarrow (r_0,r_1,\cdots,r_N) - (0,1,\cdots,N-1,1).
\end{eqnarray}

Further, using $(\ref{16})$ and $(\ref{21})$ it is straightforward (though somewhat tedious) to confirm that an arbitrary $U(F)$ given by $(\ref{16a})$ acts on $|m,r\rangle  \langle m,r|$ as follows: 
\begin{equation}
U(F) \colon |m,r\rangle\langle m,r|  \to
\begin{cases}
 \left|\frac{\alpha m+\beta}{\gamma m +\delta},\frac{r}{\gamma m +\delta}\right\rangle  \left\langle \frac{\alpha m+\beta}{\gamma m +\delta},\frac{r}{\gamma m +\delta}\right| \quad &\text{$m\neq \infty$, $\gamma m +\delta\neq 0$}
\\ \vphantom{\Bigl<}
|\infty,-\gamma r\rangle\langle \infty,-\gamma r| \quad & \text{$m\neq \infty$, $\gamma m+\delta = 0$}
\\
\left|\frac{\alpha}{\gamma},\frac{r}{\gamma}\right\rangle  \left\langle \frac{\alpha}{\gamma},\frac{r}{\gamma}\right| \quad & \text{$m=\infty$, $\gamma\neq 0$}
\\ \vphantom{\Bigl<}
|\infty,\delta r \rangle \langle \infty,\delta r| \quad &\text{$m=\infty$, $\gamma = 0$}.
\end{cases}
\label{e32} 
\end{equation}
On $r_m$ these translate into
\begin{equation}
r_m \to
\begin{cases}
-(\gamma m-\alpha)r_{\frac{\beta-\delta m}{\gamma m-\alpha}} \qquad & \text{$m\neq N$, $\gamma m-\alpha \neq 0$}
\\
\frac{1}{\gamma}r_N \qquad & \text{$m\neq N$, $\gamma m-\alpha =0$}
\\
 -\gamma r_{\frac{-\delta}{\gamma}} \qquad &  \text{$m= N$, $\gamma\neq 0$}
 \\
\frac{1}{\alpha}r_N \qquad &  \text{$m= N$, $\gamma = 0$}
\end{cases}
\end{equation}
On the set of MUBs the action is that of a M\"obius transformation. In fact the set of MUBs 
can be regarded as a projective line at infinity, added to the affine plane.
%In particular,
%\begin{equation}
%V:{\begin{array}{l} |m,r\rangle  \langle m,r|\rightarrow |m+1,r\rangle  \langle m+1,r|,~~m=0,1,\cdots,N-1\\
%\\
%|\infty,r\rangle  \langle \infty,r|\rightarrow |\infty,r\rangle  \langle \infty,r|
%\end{array}
%\label{e35} \end{equation}
%\begin{equation}
%W:{\begin{array}{l} |0,r\rangle  \langle 0,r|\rightarrow |\infty,r\rangle  \langle \infty,r|,\\
%   |m,r\rangle  \langle m,r|\rightarrow |-m^{-1},-m^{-1}r\rangle  \langle -m^{-1},-m^{-1}r|,~~m=1,\cdots,N-1\\
%   |\infty,r\rangle  \langle \infty,r|\rightarrow |0,-r\rangle  \langle 0,-r|
%\end{array}}
%\label{e36} \end{equation}
%which translate into
%\begin{eqnarray} 
%V &:& (r_0,r_1,\cdots,r_N) \rightarrow (r_{N-1},r_0,\cdots,r_{N-2},r_N)  \\
%W &:&(r_0,r_1,\cdots,r_N) \rightarrow (-r_N,\cdots,mr_{-m^{-1}},\cdots,r_0) 
%\end{eqnarray}

These actions can be written out as matrices acting on $\mathbf{r}\equiv (r_0,\cdots,r_N)$. Thus, for instance, for $N=3$ we find for the generators of $\SL(2, \mathbb{Z}_3)$  
\begin{align}
V\colon \mathbf{r} \rightarrow \mathbf{r}^\prime &=\mathcal{U}(g_1)\mathbf{r},\\
W\colon \mathbf{r}\rightarrow \mathbf{r}^\prime &=\mathcal{U}(g_2)\mathbf{r},
\end{align}
where
\begin{equation}
\mathcal{U}(g_1)=\begin{pmatrix}0&0&1&0\\1&0&0&0\\0&1&0&0\\0&0&0&1\end{pmatrix}; \qquad
\mathcal{U}(g_2)=\begin{pmatrix}0&0&0&2\\0&0&1&0\\0&2&0&0\\1&0&0&0\end{pmatrix}.
\end{equation}

\noindent Note that these matrices have determinant equal to one and hence belong to 
the group $\SL(N+1,\mathbb{Z}_N)$. 

Next we develop a similar notation for the affine planes.
As discussed earlier, the set of phase point operators obtained by applying all elements of the Weyl-Heisenberg group to a fixed phase point operator defines an affine plane containing that phase point operator. This operation decomposes the set of $N^{N+1}$ phase point operator into $N^{N-1}$ disjoint subsets---the affine planes, containing $N^2$ points each. To 
develop a useful representation for the affine planes as cosets in  $V_{N+1}(\mathbb{Z}_N)$ we 
notice that the action of the Weyl-Heisenberg group on the phase point operators consists in adding vectors 
which are linear combinations of $e_0$ and $e_1$, 
where
\begin{eqnarray}
e_0&=& (1,1,\cdots,1,0), \\
e_{1}&=& (0,1,2,\cdots, N-1,1).
\end{eqnarray}
Further, under the action of $V$ and $W$ these transform into each other:
\begin{eqnarray} 
V &:&  e_0\rightarrow e_0,~~ e_{1} \rightarrow e_1-e_0,  \\
W &:&  e_0\rightarrow e_1,~~ e_{1} \rightarrow -e_0. 
\end{eqnarray}

It therefore proves convenient to choose a basis which contains these two. Any choice would do. We choose the remaining ones as:
\begin{eqnarray}
e_k&=& (0,1^k,2^k,\cdots,(N-1)^k,0);~~k=2,\cdots, N-1,  \nonumber
\\
e_{N}&=& (0,0,0,\cdots,0,1).
\end{eqnarray}
That this set of vectors is indeed a linearly independent set can easily be checked
using the properties of Van der Monde determinants. 
We shall denote the components of a vector $(r_0,r_1,\cdots,r_N)$ in the e-basis as 
 $[\alpha_0,\alpha_1,\cdots,\alpha_N]$ and the invertible matrix relating the two by $S$:
\begin{equation}
{\bf r}=S{ \boldsymbol{\alpha}}; \ S=(e_{0}^T, e_{1}^T,\cdots,e_{N}^T).
\label{S}
\end{equation} 
Under the action of the Weyl-Heisenberg group, the components in the e-basis have rather simple properties 
\begin{eqnarray} D(1,0) &:& [\alpha_0,\alpha_1,\cdots,\alpha_N] \rightarrow [\alpha_0+1,\alpha_1,\cdots,\cdots,\alpha_N], \label{65} \\
 D(0,1) &:& [\alpha_0,\alpha_1,\cdots,\alpha_N]\rightarrow [\alpha_0,\alpha_1-1,\cdots,\alpha_N].\label{66}  
\end{eqnarray}

\noindent Thus the collection $[\alpha_0,\alpha_1,\cdots,\alpha_N]$, with $\alpha_0$ 
and $\alpha_1$ taking all values in $\mathbb{Z}_N$ with the rest fixed, defines an 
affine plane. We label it by the $N-2$ coordinates 
$[\alpha_2,\cdots,\alpha_{N-1},\alpha_N]$. 
 
The action of $V$ and $W$ is %given below: 
\begin{align}
V\colon \boldsymbol{\alpha}\rightarrow \boldsymbol{\alpha}^\prime &=S^{-1}\mathcal{U}(g_1){S}\boldsymbol{\alpha},\\
W\colon \boldsymbol{\alpha}\rightarrow \boldsymbol{\alpha}^\prime &=S^{-1}\mathcal{U}(g_2)S\boldsymbol{\alpha}.
\end{align}
where the matrix $S$ is given in $(\ref{S})$.

If one is only interested in the actions of $V$ and $W$ on the affine planes then it suffices to look at the actions of the matrices $S^{-1}\mathcal{U}(g_1)S$ and $S^{-1}\mathcal{U}(g_2)S$ on the $N-1$ $\alpha$'s omitting $\alpha_0$ and $\alpha_1$. 
%As with phase point operators, one finds t
These matrices have determinant equal to one and hence belong to the group $\SL(N-1,\mathbb{Z}_N)$. 

To summarise, we have a representation of the $\SL(2,\mathbb{Z}_N)$ action on the phase space through matrices belonging to $\SL(N+1,\mathbb{Z}_N)$ ( $\SL(N-1,\mathbb{Z}_N)$ ) acting on  an $N+1$ ($N-1$) dimensional vector space over $\mathbb{Z}_N$ whose elements are in one to one correspondence with the phase point operators (affine planes). A natural question to ask is whether or not one has an analogue of the symplectic form left invariant by $\SL(2,\mathbb{Z}_N)$ action on the phase space. 
To pursue this question we look for bilinear quadratic forms of the type $\Omega({\bf r}^{\prime},{\bf r})\equiv{\bf r}^{\prime T}~\Omega~{\bf r}$, left invariant by the $U(g_1), U(g_2)$ action on the ${\bf r}$'s, i.e we look for matrices $\Omega$ such that
\begin{equation}
\mathcal{U}^T (g_1)\Omega \mathcal{U} (g_1)=\Omega;~~ \mathcal{U}^T (g_2)\Omega \mathcal{U} (g_2)=\Omega.
\end{equation} 
The antisymmetric matrix $\Omega$ turns out to be %the following:
\begin{equation}
\Omega=  \left( \begin{array}{cc} {\bf q} & \mathbb{I} \\ -\mathbb{I}^T & 0 \end{array} \right),  
\end{equation}
where $\mathbb{I}$ denotes an $N$-dimensional column with all entries equal to $1$ and 
${\bf q}$ is an $N\times N$ antisymmetric matrix with non-vanishing matrix elements
\begin{equation}
q_{ij}= \frac{1}{j-i}.
\end{equation}
In the e-basis, the symplectic form $\Omega$ has a very simple structure:
\begin{equation}
\Omega(e_i,e_j)=\begin{cases} 
0 \qquad & \text{$i$ or $j=0$ or $1$},\\
             j\delta_{i+j,0} \qquad & 2\leq i,j\leq N-2,\\
           -1 \qquad & \text{ $i=N-1$ and  $j=N$},
\end{cases} 
\end{equation}
which in turn suggest a basis in which $\Omega$ has the canonical form. Thus,
 for instance, for $N=5$ the basis in which $\Omega$ has the canonical form
\begin{equation}
  \left( \begin{array}{cc} {\bf 0} & \mathbb{I} \\ -\mathbb{I} & {\bf 0}\end{array} \right),  
\end{equation}
is $$e_2,e_4,\frac{1}{3}e_3,-e_5 \ .$$
Similarly for $N=7$ we have  $$e_2,e_4,e_6,\frac{1}{5}e_5,\frac{1}{3}e_3,-e_7 \ .$$

\section{Orbits under the Clifford Group}

Having learnt how the generators of the Clifford group act on phase point operators and affine planes, our next task is to 
investigate how they arrange themselves into orbits under the action of this group.   As is well known, to generate an orbit of a group $\mathcal{G}$ acting on a set  we need to pick an element of the set and apply all elements of the group to it. We then pick another element in the set not contained in the earlier set and generate its orbit and so on until all elements of the set are exhausted. The quantities of interest are then (a) the number of orbits (b) the size of each orbit. The total number of orbits can be calculated using the Burnside Lemma:
\begin{equation}
{\rm Number~ of~ orbits}=\frac{1}{|\mathcal{G}|}\sum_{g\in \mathcal{G}} \phi(g),
\end{equation}
where $|\mathcal{G}|$ denotes the order of the group and $\phi(g)$ the number of points of the set left fixed by the action of $g$. As $g$ and the elements related to it by conjugation have the same fixed points, to calculate the number of orbits we need to know the number of elements in each conjugacy class and the number of  points left fixed by a representative in each class. 
As to the size of the orbits, it is given by the ratio of the order of the group to the order of the stability group of the starting element of the set from which the orbit is built up by group action. Knowing the order of the subgroups, one can then deduce the possible sizes of the orbits. In summary, to answer the questions pertaining to the orbits, we need 
\begin{enumerate}
\item knowledge of the conjugacy classes and the number of elements therein 
\item expressions for the class representatives in terms of the generators 
\item orders of the cyclic subgroups generated by of elements of a conjugacy class. 
\end{enumerate}
In the following we collect together some relevant facts from the literature 
\cite{12} concerning these aspects for the group $\SL(2,\mathbb{Z}_N)$ (also see the Appendix). 

 In our discussion we will need to distinguish the sets 
\begin{align}
Q &= \{x \in \mathbb{Z}^{*}_N:  x = y^2 \text{ for some } y\in\mathbb{Z}^{*}_N\} \\
\bar{Q} & = \{x\in\mathbb{Z}^{*}_N: x \notin Q\}
\end{align}
where $\mathbb{Z}^{*}_N$ is the set of non-zero elements of $\mathbb{Z}_N$ (so $Q$ is the set of quadratic residues and $\bar{Q}$ is the set of non-quadratic residues~\cite{Hardy}).

It is immediate that elements of $\SL(2,\mathbb{Z}_N)$ having different trace must belong to different conjugacy classes.  It turns out~\cite{12} that there is in fact exactly one conjugacy class for each value of the trace, except when the trace $=\pm 2$ in which case there are three.  This gives us $N+4$ conjugacy classes:  the $N$ ``standard'' classes
\begin{equation}
C_{t} = \left[ \begin{pmatrix} 0 & -1 \\ 1 & t\end{pmatrix} \right]
\end{equation}
with $t=0,1,\dots (N-1)$ and the four additional classes
\begin{equation}
\bar{C}_{\pm 2}= \left[ \begin{pmatrix} 0 & -\frac{1}{\nu}  \\ \nu & \pm 2 \end{pmatrix} \right] \qquad \text{and} \qquad D_{\pm 2} = \left[ \begin{pmatrix} \pm 1 & 0  \\ 0 & \pm 1 \end{pmatrix} \right]
\end{equation}
where $\nu$ is any fixed element of $\bar{Q}$.  In these expressions $[F]$ denotes the conjugacy class containing $F$.  For more details see Theorem~\ref{thm:SLConjClass} in the Appendix.

For the standard classes we have
\begin{equation}
|C_t| = 
\begin{cases}
N(N+1) \qquad & \text{if $t^2-4 \in Q$}\\
N(N-1) \qquad & \text{if $t^2-4 \in \bar{Q}$}\\
\frac{1}{2} (N^2-1)\qquad & \text{if $t^2-4 = 0$}
\end{cases}
\end{equation}
where the notation $|S|$ means ``number of elements in the set $S$''.   For the additional classes we have
 $|\bar{C}_{\pm 2}| = (N^2-1)/2$ and $|D_{\pm 2}| =1$.  It may be worth noting that the classes $C_t$ for which $t^2-4\in Q$ are precisely the ones whose elements are diagonalizable (apart from the two classes $D_{\pm 2}$).

For trace $= 0, \pm 1, \pm 2$  it is easy to calculate the orders of the cyclic subgroups generated by the elements in each conjugacy class for arbitrary $N$.   They are listed in Table~\ref{tb:genOrd} 
\begin{table}[h]
\begin{center}
\renewcommand{\arraystretch}{1.3}
\begin{tabular}{| c | c | c| c | c |c |c | c | c | c  | }\hline
\te{2}{class} & 
  \te{0.7}{$C_0$ }&
  \te{0.7}{$C_1$ }&
  \te{0.7}{$C_{-1}$ }&
  \te{0.7}{$C_2$}&  
  \te{0.7}{$\bar{C}_2$ }& 
  \te{0.7}{$D_2$ }&  
  \te{0.7}{$C_{-2}$}  &  
  \te{0.7}{$\bar{C}_{-2}$} &  
  \te{0.7}{$D_{-2}$} 
  \\ \hline
\te{2}{order }& 
  \te{0.7}{$4$ }&
  \te{0.7}{$6$ }&
  \te{0.7}{$3$ }&
 \te{0.7}{$N$} &  
 \te{0.7}{$N$} &    
 \te{0.7}{$1$} &    
 \te{0.7}{$2 N$ }&    
 \te{0.7}{$2 N$ }&    
 \te{0.7}{$2 $ } \\
\hline
\end{tabular}
\vspace{0.2 cm}
\caption{The order of the cyclic subgroups generated by elements of some special 
conjugacy classes, for all $N$.}
\label{tb:genOrd}
\end{center}
\end{table}
(to derive the result for elements of $C_{\pm 2}$, $\bar{C}_{\pm 2}$ consider matrices of the form $\left(\begin{smallmatrix} \pm 1 & k \\0& \pm 1\end{smallmatrix}\right)$).
For the remaining conjugacy classes the orders of the cyclic subgroups depend
on number theoretical details. We have worked out the orders for $N \leq 19$.  They are listed in Table~\ref{tb:ordNLs19}.
\begin{table}[h]
\begin{center}
\renewcommand{\arraystretch}{1.3}
\begin{tabular}{| c| c| c| c| c| c| c| c| c| c| c| c| c| c| c|} \hline 
&\te{   0.43}{$ \te{   0.43}{$C_3$} $}&\te{   0.43}{$ C_4 $}&\te{   0.43}{$ C_5$}&\te{   0.43}{$C_6$}&\te{   0.43}{$C_7$}&\te{   0.43}{$C_8$}&\te{   0.43}{$C_9$}&\te{   0.43}{$C_{10}$}&\te{   0.43}{$C_{11}$}&\te{   0.43}{$C_{12}$}&\te{   0.43}{$C_{13}$}&\te{   0.43}{$C_{14}$}&\te{   0.43}{$C_{15}$}&\te{   0.43}{$C_{16}$}
\\ 
\hline 
\te{1.2}{$N=7 $}&\te{   0.43}{$ 8 $}&\te{   0.43}{$ 8 $}&
&& && && && && &
\\
\hline 
\te{1.2}{$N=11 $}&\te{   0.43}{$ 5 $}&\te{   0.43}{$ 10 $}&\te{   0.43}{$12 $}&\te{   0.43}{$12 $}&\te{   0.43}{$5 $}&\te{   0.43}{$10 $}&
&& && && &
\\ 
\hline 
\te{1.2}{$N=13 $}&\te{   0.43}{$ 14 $}&\te{   0.43}{$ 12 $}&\te{   0.43}{$14 $}&\te{   0.43}{$14 $}&\te{   0.43}{$7 $}&\te{   0.43}{$7 $}&\te{   0.43}{$12 $}&\te{   0.43}{$7 $}&
&& && &
\\
 \hline 
\te{1.2}{$N=17 $}&\te{   0.43}{$ 18 $}&\te{   0.43}{$18 $}&\te{   0.43}{$ 16$}&\te{   0.43}{$ 8$}&\te{   0.43}{$ 9$}&\te{   0.43}{$ 16$}&\te{   0.43}{$16 $}&\te{   0.43}{$18 $}&\te{   0.43}{$8 $}&\te{   0.43}{$16 $}&\te{   0.43}{$9 $}&\te{   0.43}{$ 9$}&
&
\\ 
\hline 
\te{1.2}{$N=19 $}&\te{   0.43}{$ 9 $}&\te{   0.43}{$ 5 $}&\te{   0.43}{$10 $}&\te{   0.43}{$20 $}&\te{   0.43}{$9 $}&\te{   0.43}{$20 $}&\te{   0.43}{$9 $}&\te{   0.43}{$18 $}&\te{   0.43}{$20 $}&\te{   0.43}{$18 $}&\te{   0.43}{$20 $}&\te{   0.43}{$ 5 $}&\te{   0.43}{$10 $}&\te{   0.43}{$18$}
\\ 
\hline 
\end{tabular}
\vspace{0.2 cm}
\caption{The order of the cyclic subgroups generated by elements of all the conjugacy 
classes, for $N \leq 19$.}
\label{tb:ordNLs19}
\end{center}
\end{table}
Since we know how many elements there are in each conjugacy class, the number of distinct cyclic subgroups can easily be worked out. 
In particular we find that there is a unique subgroup of order $2$, generated
by the element $- I \in D_{-2}$;  $N+1$ (Sylow) subgroups of order $N$, generated by 
elements of $C_{2}$, $\bar{C}_{2}$; and $N+1$ subgroups of order $2N$ generated by elements 
of $C_{-2}$, $\bar{C}_{-2}$.

We can now compute 
quantities of interest concerning the orbits generated by Clifford group action on phase point operators and affine planes.   Using the results in the previous section we have done this for   $N=3,5, 7$.  The results are listed in Table~\ref{tb:fixedEtc}.
\begin{table}
\begin{center} 
\begin{tabular}{| c| c| c| c| c| c| c|} \hline
\te{1.4}{N} & \te{1.4}{Class} & \te{1.4}{number of elements} & \te{1.4}{order of cyclic subgroup} & \te{1.4}{fixed points in plane} & \te{1.4}{fixed $A$s}& \te{1.4}{fixed planes} 
\\ \hline 
$3$ & $C_0$ & 6 & 4 & 1 & 1 & 1 \\
& $C_{1}$ & 4 & 6 & 1 & 1 & 1 \\
& $\bar{C}_{1}$ & 4 & 6 & 1 & 1 & 1 \\ 
& $D_{1}$ & 1 & 2 & 1 & 1 & 1 \\  
& $C_2$ & 4 & 3 & 3 & 9 & 3 \\
& $\bar{C}_2$ & 4 & 3 & 3 & 9 & 3 \\ 
 & $D_2$ & $1$ & $1$ & $9$ & $81$ & $9$ \\ 
\hline 
5 & $C_{0}$ & 30 & 4 & 1 & 1 & 1 \\ 
& $C_{1}$ & 20 & 6 & 1 & 1 & 1 \\ 
 & $C_2$ & 12 & 5 & 5 & 25 & 5 \\ 
& $\bar{C}_2$ & 12 & 5 & 5 & 25 & 5 \\ 
 & $D_2$ & 1 & 1 & 25 & 15625 & 625 \\ 
& $C_{3}$ & 12 & 10 & 1 & 1 & 1 \\
& $\bar{C}_{3}$ & 12 & 10 & 1 & 1 & 1 \\
& $D_{3}$ & 1 & 2 & 1 & 1 & 1 \\ 
& $C_{4}$ & 20 & 3 & 1 & 25 & 25 \\ 
\hline 
7 & $C_{0}$ & 42 & 4 & 1 & 1 & 1 \\ 
& $C_{1}$ & 56 & 6 & 1 & 1 & 1 \\ 
& $C_2$ & 24 & 7 & 7 & $7^2$ & 7 \\ 
& $\bar{C}_2$ & 24 & 7 & 7 & $7^2$ & 7 \\ 
 & $D_2$ & 1 & 1 & $7^2$ & $7^8$ & $7^6$ \\ 
& $C_{3}$ & 42 & 8 & 1 & 1 & 1 \\ 
& $C_4$ & 42 & 8 & 1 & 1 & 1 \\ 
& $C_{5}$ & 24 & 14 & 1 & 1 & 1 \\
& $\bar{C}_{5}$ & 24 & 14 & 1 & 1 & 1 \\ 
& $D_{5}$ & 1 & 2 & 1 & 1 & 1 \\
& $C_6$ & 56 & 3 & 1 & $7^2$ & $7^2$ \\ \hline
\end{tabular}
\vspace{0.2 cm}
\caption{Fixed points under the action of the elements in $SL(2, Z_N)$, $N \leq 7$. 
The action can be on an affine plane, on the set of phase point operators, and on 
the set of affine planes.}
\label{tb:fixedEtc}
\end{center}
\end{table}
Some things can be said in general. Notably cyclic subgroups of even order have only one 
fixed point. It appears that cyclic subgroups of an odd order $m$ always leave exactly 
$N^k$ affine planes fixed, where $N + 1 = km + m'$, $m' <  m$. (We proved this for subgroups 
of order 3, and checked it in all cases for $N \leq 19$.)

On the affine plane the number of orbits is always $2$. On the set of affine planes the number of orbits can be found using Burnside's lemma.  The results for $N\le 11$ are listed in Table~\ref{tb:noOrbits}. 
\begin{table}[h]
\begin{center}
\renewcommand{\arraystretch}{1.3}
\begin{tabular}{|c|c|c|c|c|} \hline 
\te{2}{$N$} & \te{1.5}{$3$} & \te{1.5}{$5$} &\te{1.5}{$7$}&\te{1.5}{$11$}\\ \hline 
No.\ of orbits & $2$ & $11$ &$360$& $19650810$\\ \hline
\end{tabular}
\vspace{0.2 cm}
\caption{Number of orbits of $SL(2, Z_N)$ when acting on the set of affine planes.}
\label{tb:noOrbits}
\end{center}
\end{table}
They should be compared with the estimate $N^{N-1}/N^3$ which for $N=5,7,11$ gives
$5,343,19487171$ respectively. 
The estimate is rather good. 
For $N=3$ the sizes of the orbits are $1+8=3^2$, and for $N = 5$ they are $1+24+40+40+40+40+40+40+120+120+120=5^4$. For $N=7$ the results are too numerous to be given here. As expected, each term in the sum is a divisor of the order $N(N^2-1)$ of the group. 
We observe that for all $N$ there is a unique singlet, and a unique $(N^2-1)$-plet consisting 
of affine planes left invariant by some Sylow subgroup..

\section{Spectra of Phase Point Operators, Complex Conjugation and the Extended Clifford Group}
We now turn to the spectra of the phase point operators. Obviously those in the same affine plane have the same spectra and hence we need to focus only on the affine planes. For 
$N=3$ one finds two distinct spectra: $(1,1,-1)$ and $(\Phi,1-\Phi,0)$ where $\Phi=
(1+\sqrt{5})/2=1.61803..$, the Golden Ratio. No surprises here. 

For $N=5$ one finds only 9 distinct spectra as against the naive expectation of 11, as listed in Table~\ref{tb:spectra}.
\begin{table}[h]
\begin{center}
\renewcommand{\arraystretch}{1.4}
\begin{tabular}{|c|c|}
\hline
spectrum & no. of occurrences \\
\hline
$ \{-1.,-1.,1.,1.,1.\} $ & 1  \\ \hline
$ \{-1.,-0.61803,0,1.,1.61803\}$ & 24  \\ \hline
 $\{-0.94658,-0.5169,-0.18438,0.93842,1.70944\}$ & 120  \\ \hline
 $\{-0.90932,-0.48701,0,0.46853,1.9278\}$ & 120  \\ \hline
 $\{-0.90039,-0.64018,-0.14531,1.06785,1.61803\}$ & 40  \\ \hline
 $\{-0.83726,-0.58152,-0.09576,0.6287,1.88584\}$ & 120  \\ \hline
 $\{-0.83607,-0.81,0,1.05469,1.59139\}$ & 80  \\ \hline
 $\{-0.79859,-0.36221,0,0.10661,2.05419\}$ & 80  \\ \hline
 $\{-0.70281,-0.61803,-0.13294,0.48666,1.96712\}$ & 40 \\ \hline
\end{tabular}
\vspace{0.2 cm}
\caption{Possible spectra of the phase point operators in an affine plane, $N = 5$.}
\label{tb:spectra}
\end{center}
\end{table}
For $N=7$ one gets 210 distinct spectra (too numerous to reproduce here). In both cases one finds that sizes of the spectral orbits do not divide the order of the group 
but rather twice the order. 
This leads one to suspect that as far as the spectra are concerned, the relevant group is not the Clifford group but rather an augmented Clifford group.  
The extra symmetry is provided by the operation of complex conjugation which also maps the MUB's into themselves. The Extended Clifford Group~\cite{10}
includes the complex conjugation operation and has twice as many elements as the Clifford group. This is consistent with the sizes of the spectral orbits found. With this in mind we investigate and encode the action of complex conjugation on the MUB's, so as to have explicit actions of all the operations on the MUB's relevant for discussing the spectra of the phase point operators which in turn enable us to classify 
distinct Wigner distributions.

Complex conjugation, denoted by $C$, is represented by an anti-unitary operator. Under $C$ 
the MUB's transform into each other 
as follows: 
\begin{align}
C\colon |m,r\rangle & \to
\begin{cases}
|-m, r\rangle \qquad & m \neq \infty \\
|\infty, -r\rangle \qquad & m = \infty
\end{cases}
\\
\intertext{which implies}
C\colon r_m & \to
\begin{cases}
r_{-m} \qquad & m = 0, 1, \dots, N-1 \\
-r_N \qquad & m = N
\end{cases}
\\
\intertext{and}
C\colon e_k & \to (-1)^k e_k \qquad k = 0, 1, \dots, N
\end{align}
These actions, when translated into actions on the $\alpha$'s, have 
the pleasant feature that under complex conjugation the $\alpha$'s alternately change signs: 
\begin{equation}
{\rm C}:[\alpha_0,\alpha_1,\cdots,\alpha_N]\rightarrow
[\alpha_0,-\alpha_1,\alpha_2,\cdots,-\alpha_N].
\end{equation}
Thus complex conjugation acts on the phase point operators $[\alpha_0,\alpha_1,\cdots,\alpha_N]$  by a diagonal matrix ${\rm diag}(1,-1,1\cdots,-1)$, 
with determinant $-1$ if $N=4k+1$ and $1$ if $N=4k+3$ . 
On the affine planes $[\alpha_2,\alpha_{3},\cdots,\alpha_N]$, the action is again 
by a diagonal matrix but with determinant  $1$ if $N=4k+1$ and $-1$ if $N=4k+3$. 

As described in ref.~\cite{10} the extended Clifford group is obtained by admitting matrices in $\ESL(2,\mathbb{Z}_N)$ (the group consisting of $2\times 2$ matrices with entries in $\mathbb{Z}_N$ and   determinant $=\pm 1$).  The conjugacy classes of $\ESL(2,\mathbb{Z}_N)$ are described in the Appendix to this paper.  Unlike  $\SL(2,\mathbb{Z}_N)$ there is a distinction between the cases $N = 1 \text{ (mod $4$)}$ and $N = 3 \text{ (mod $4$)}$.    It turns out that for all values of $N$ we have the $2 N$ ``standard'' classes 
\begin{equation}
C_{\Delta, t} = \left[\begin{pmatrix} 0 & -\Delta \\ 1 & t\end{pmatrix} \right]
\end{equation}
where $\Delta = \pm 1$ and $t = 0, 1, \dots , (N-1)$, and the $2$ additional classes
\begin{equation}
D_{\pm 2} = \left[ \begin{pmatrix} \pm 1 & 0 \\ 0 & \pm 1 \end{pmatrix} \right]
\end{equation}
If $N= 3 \text{ (mod $4$)}$ these are the only conjugacy classes (see Theorem~~\ref{thm:ESLconj3mod4} in the Appendix).  If, however, $N = 1 \text{ (mod $4$)}$ there are another $6$ conjugacy classes, which we denote $\bar{C}_{\pm 2}$, $\bar{C}_{\pm 2 i}$, $D_{\pm 2 i}$ (see Theorem~\ref{thm:ESLconj1mod4} in the Appendix).

All elements of $ESL$ with determinant $-1$ are represented on the symplectic vector 
space of affine planes, $V_{N-1}(\mathbb{Z}_N)$, by matrices $\mathcal{A}$ that obey 

\begin{equation} \mathcal{A}^T\Omega \mathcal{A} = - \Omega \ . \end{equation}

\noindent They are therefore anti-canonical transformations. In Table~\ref{tb:fixedESL}  
\begin{table}[!ht]
\begin{center}
\begin{tabular}{| c| c| c| c| c| c| c|}  \hline
\te{1.4}{N} & \te{1.4}{Class} & \te{1.4}{number of elements} & \te{1.4}{order of cyclic subgroup} & \te{1.4}{fixed points in plane} & \te{1.4}{fixed $A$s}& \te{1.4}{fixed planes} 
\\ \hline 
3 & $C_{-1,0}$ & 12 & 2 & 3 & $3^2$ & 3 \\ 
& $C_{-1,1}$ & 6 & 8 &  1 & 1 & 1 \\
& $C_{-1,2}$ & 6 & 8 &  1 & 1 & 1 \\ 
\hline
5 & $C_{-1,0}$ & 30 & 2 & 5 & $5^3$ & $5^2$ \\ 
& $C_{-1,1}$ & 12 & 20 & 1 & 1 & 1 \\
& $\bar{C}_{1}$ & 12 & 20 & 1 & 1 & 1 \\
& $D_{1}$ & 1 & 4 & 1 & 1 & 1 \\
& $C_{-1,2}$ & 20 & 20 & 1 & 1 & 1 \\ 
& $C_{-1,3}$ & 20 & 12 & 1 & 1 & 1 \\
& $C_{-1,4}$ & 12 & 20 & 1 & 1 & 1 \\
& $\bar{C}_{4}$ & 12 & 20 & 1 & 1 & 1 \\
& $D_{4}$ & 1 & 4 & 1 & 1 & 1 \\
\hline 
7 & $C_{-1,0}$ & 56 & 2 & 7 & $7^4$ & $7^3$ \\
& $C_{-1,1}$ & 42 & 16 & 1 & 1 & 1 \\
& $C_{-1,2}$ & 56 & 6 & 1 & 7 & 7 \\
& $C_{-1,3}$ & 42 & 16 & 1 & 1 & 1 \\
& $C_{-1,4}$ & 42 & 16 & 1 & 1 & 1 \\
& $C_{-1,5}$ & 56 & 6 & 1 & 7 & 7 \\
& $C_{-1,6}$ & 42 & 16 & 1 & 1 & 1 \\ 
\hline
\end{tabular} 
\vspace{0.2 cm}
\caption{Fixed points under the action of elements in $ESL(2, Z_N)$ with 
determinant $-1$, $N \leq 7$.}
\label{tb:fixedESL} 
\end{center}
\end{table}
we list the number of fixed points for the conjugacy classes of $\ESL(2,\mathbb{Z}_N)$ 
with determinant $-1$. For all $N$ one can show that elements of order 2 leave 
$N^{(N-1)/2}$ affine planes fixed. Because the transformations are anti-canonical it 
follows that these fixed planes form Lagrangian subspaces of $V_{N-1}(\mathbb{Z}_N)$. 
Furthermore we believe 
that if an element with determinant $-1$ squares to an element leaving $N^{2k}$ $A$s 
fixed, then that element itself leaves $N^k$ $A$s fixed; we checked this for $N \leq 11$.

Again we can use Burnside's lemma to compute the number of orbits under $\ESL(2,\mathbb{Z}_N)$.  The results for $N\le 11$ are listed in Table~\ref{tb:noOrbitsE}.
\begin{table}
\begin{center}
\renewcommand{\arraystretch}{1.3}
\begin{tabular}{|c|c|c|c|c|} \hline 
\te{2}{$N$} & \te{1.5}{$3$} & \te{1.5}{$5$} &\te{1.5}{$7$}&\te{1.5}{$11$}
\\ \hline 
No.\ of orbits & 2 & 9 &210& 9833460\\ \hline
\end{tabular}
\vspace{0.2 cm}
\caption{Number of orbits of $ESL(2, Z_N)$ when acting on the set of affine planes.}
\label{tb:noOrbitsE}
\end{center}
\end{table}
For large $N$ the number approaches the estimate $N^{N-1}/2N^3$. For $N \leq 7$ we know the number 
of distinct spectra, and we find complete agreement with the number of orbits.

\section{Summary}
\bigskip
\noindent

To summarise:
\begin{itemize}
\item For $N$ an odd prime, a compact, mathematica friendly way of representing phase point operators and affine planes as vectors and cosets in an $N+1$ dimensional vector space over $Z_N$ is given. 
\item In this representation the $\SL(2,Z_N)$ actions on phase point operators %turn out to be 
actions by matrices in $\SL(N+1,Z_N)$. 
\item A convenient basis, factoring out the Heisenberg-Weyl group actions, is introduced in which $\SL(2,Z_N)$ actions on the affine planes can be described in terms of matrices in $\SL(N-1,Z_N)$.
\item The action of complex conjugation on the phase point operators %turns out to be 
is by matrices in   $\GL(N+1,Z_N)$ with determinant $-1$ if $N=4k+1$ and $1$ if $N=4k+3$.
\item The action of complex conjugation on the affine planes %, however, turns out to be again 
is by  matrices in   $\GL(N-1,Z_N)$ with determinant= $1$ if $N=4k+1$ and $-1$ if $N=4k+3$.
\item The $(N-1)$-dimensional vector space of affine planes is symplectic. Elements of 
$ESL(2, Z_N)$ with determinant $1$ preserve the symplectic form, elements with 
determinant $-1$ are anti-canonical. 
\item For $N \leq 7$ the number of distinct spectra in the phase point operators 
equals the number of orbits under the group $\ESL(2, \mathbb{Z}_N)$.

\end{itemize}

The work reported here suggests a number of avenues for future research.   In the first place one would like to know if the degeneracies of the phase point operators are completely determined by $\ESL(2, \mathbb{Z}_N)$ for every odd prime dimension $N$.  Beyond that one would like to know if the result generalizes to the case of arbitrary odd prime \emph{power} dimension.    Finally, there is the question of even prime power dimensions.  In that connection let us observe that Gibbons et al \cite{5} find 4 
similarity classes when $N = 4$.  However, when we computed the spectra we found only the $3$ distinct spectra tabulated in  Table~\ref{tb:spectra2}.
\begin{table}[h]
\begin{center}
\renewcommand{\arraystretch}{1.4}
\begin{tabular}{|c|c|}
\hline
spectrum & no. of occurrences \\
\hline
 $\{-0.896802,-0.14204,0.278768,1.76007\}$ & 384  \\ \hline
$ \{-0.866025,-0.5,0.866025,1.5\}$ & 320  \\ \hline
 $ \{-0.5,-0.5,0.133975,1.86602\} $ & 320  \\ \hline
\end{tabular}
\vspace{0.2 cm}
\caption{Possible spectra of the phase point operators in an affine plane, $N = 4$.}
\label{tb:spectra2}
\end{center}
\end{table}
 The cause of this additional degeneracy cannot be the same as the cause in odd prime dimension since when $N$ is an even prime power the fact that $-1 = 1 \text{ (mod $2$)}$ means that $\SL(2, \mathbb{F}_N)= \ESL(2, \mathbb{F}_N)$.  There must therefore be some other explanation, which it would be interesting to investigate.

\vspace{15mm}

{\bf Acknowledgements:}

\

\noindent We thank Ernesto Galv\~ao for sharing some preliminary results in this direction. 
Our work was financially supported by the Wenner-Gren Foundations, as well as by the 
Swedish Research Council. One of us (SC) also wishes to thank the Leverhulme
Trust for a Visiting Professorship at the Imperial College London where a
part of this work was done. 

\appendix

\section*{Appendix:  Conjugacy Classes of $\ESL(2,\mathbb{F}_N)$}
In this appendix we deduce the conjugacy classes for the group $\ESL(2,\mathbb{F}_N)$ where $N=p^l$ is a power of an odd prime number $p$ and  $\mathbb{F}_N$ is the Galois field~\cite{Lidl} having $N$ elements (in the main text we only need the result for the case $l=1$, $\mathbb{F}_N = \mathbb{Z}_N$; however, the result for arbitrary $l$ is no more difficult, and it may be useful for future developments).  For the sake of completeness, and the convenience of the reader,  we also describe the conjugacy classes of the group $\SL(2,\mathbb{F}_N)$ (also given in ref.~\cite{12}

Let $\mathbb{F}^{*}_N$ be the set of non-zero elements of $\mathbb{F}_N$, and define
\begin{align}
Q & = \{x \in \mathbb{F}^{*}_N \colon x = y^2 \text{ for some $y \in \mathbb{F}^{*}_N$}\}
\\
\bar{Q} & = \{x \in \mathbb{F}^{*}_N \colon x \notin Q \}
\end{align}
Let $\theta$ be a primitive element~\cite{Lidl} for $\ESL(2,\mathbb{F}_N)$.  Then it is easily seen that $\theta^s\in Q$ if and only if $s$ is even.  Consequently $Q$ and $\bar{Q}$ each contain exactly $(N-1)/2$ elements.  Also the fact that $\theta^{(N-1)/2} = -1$ means that $-1 \in Q$ if and only if $N = 1 \text{ (mod $4$)}$.  

We begin by proving three preliminary results.
\begin{lemma}
\label{lem:ResCount}
For all $x \in \mathbb{F}^{*}_N$
\begin{equation}
\left|\bar{Q}\cap \left( \bar{Q}-x\right)\right|
=
\begin{cases}
\frac{N-1}{4} \qquad & \text{if $N=1$ (mod $4$) and $x\in Q$ }\\
\frac{N-5}{4} \qquad & \text{if $N=1$ (mod $4$) and $x\in \bar{Q}$}\\
\frac{N-3}{4} \qquad & \text{if $N=3$ (mod $4$)} 
\end{cases}
\end{equation}
where the notation $\left|S\right|$ means ``number of elements in the set $S$''.
\end{lemma}
\begin{proof}
Let $x$ be any fixed element $\in \mathbb{F}^{*}_N$.  Define
\begin{align}
S_x & =\left( x \bar{Q}\right) \cap \left(\bar{Q} - x\right)
\\
T_x & = \left( x \bar{Q}\right) \cap \left(Q -x\right)
\end{align}
and let $f_x \colon \mathbb{F}^{*}_N \to \mathbb{F}^{*}_N$ be the map defined by
\begin{equation}
f_x(y) = x^2 y^{-1} .
\end{equation}
It is easily verified that $f_x$ is a bijection and (using the fact that $QQ=\bar{Q}\bar{Q} = Q$ and $Q\bar{Q} = \bar{Q}$)
\begin{equation}
f_x(S_x) = T_x \qquad f_x(T_x) = S_x.
\end{equation}
Since $S_x$, $T_x$ are bijective images of each other we deduce that $|S_x| = |T_x|$.

Suppose, now, that $N=1 \text{ (mod $4$)}$.  Then $-1\in Q$ which means $-x \notin x \bar{Q}$ for all non-zero $x$.  Consequently $\left|S_x \cup T_x\right|=\left|x\bar{Q}\right| = (N-1)/2$, implying
\begin{equation}
|S_x| = |T_x| = \frac{N-1}{4}
\end{equation}
Suppose, on the other hand, that $N=3 \text{ (mod $4$)}$.  Then $-1\in \bar{Q}$ which means $-x \in x \bar{Q}$ for all non-zero $x$.  Consequently $\left|S_x \cup T_x\right|=\left|x\bar{Q}\right| -1= (N-3)/2$, implying
\begin{equation}
|S_x| = |T_x| = \frac{N-3}{4}
\end{equation}
To complete the proof suppose, first of all, that $x\in Q$.   Then $x \bar{Q} = \bar{Q}$, implying
\begin{equation}
\left|\bar{Q}\cap \left( \bar{Q}-x\right)\right|
= \left|S_x\right|
=\begin{cases}
\frac{N-1}{4} \qquad & N = 1 \text{ (mod $4$)} \\
\frac{N-3}{4} \qquad & N = 3 \text{ (mod $4$)}
\end{cases}
\end{equation}
Suppose, on the other hand, that $x\in \bar{Q}$.   Then $x \bar{Q} = Q$, implying
\begin{align}
\left|\bar{Q}\cap \left( \bar{Q}-x\right)\right|
& = \frac{N-3}{2} - \left|Q \cap \left( \bar{Q}-x\right)\right|
\nonumber
\\
& = \frac{N-3}{2} - \left|S_x \right|
\nonumber
\\
& = \begin{cases}
\frac{N-5}{4} \quad & N = 1 \text{ (mod $4$)} \\
\frac{N-3}{4}  \quad & N = 3 \text{ (mod $4$)} 
\end{cases}
\end{align}
\end{proof}
\begin{lemma}
\label{lem:munueq}
For all $\mu \in\mathbb{F}^{*}_N$ and $\nu \in \mathbb{F}^{\vphantom{*}}_N$ there exists $q\in Q\cup \{0\}$ such that
\begin{equation}
\mu q + \nu \in Q\cup \{0\}
\end{equation}
\end{lemma}
\begin{proof}
Case 1:   $\mu \in Q$.  We have $\left| Q\cup \{0\} + \nu \right| >\left|\bar{Q}\right|$.  Consequently the set $\left( Q\cup \{0\} + \nu \right) \cap \left( Q\cup \{0\}\right)$ is non-empty.  We can therefore choose $q' \in Q\cup\{0\}$ such that $q' + \nu \in Q\cup\{0\}$.  Then $q = \mu^{-1} q'$  has the stated property.

Case 2: $\mu \in \bar{Q}$.  If $\nu = 0$ we can choose $q=0$.  If, on the other hand, $\nu \neq 0$ we note that it follows from Lemma~\ref{lem:ResCount} that $\left|\left(\bar{Q} + \nu\right)\cap \bar{Q}\right| <|\bar{Q}|$.  So there exists $q' \in\bar{Q}$ such that $q' + \nu\in Q\cup\{0\}$.  Then $q= \mu^{-1} q'$ has the stated property.  
\end{proof}
\begin{lemma}
\label{lem:simTrans}
Let 
\begin{equation}
F = \begin{pmatrix} \alpha & \beta \\ \gamma & \delta \end{pmatrix}
\end{equation}
be any matrix $\in \ESL(2,\mathbb{F}_N)$.  Let $\Delta = \det F$ and $t= \tr F$.  Then
\begin{enumerate}
\item If $t^2 - 4 \Delta\neq 0$ there exists $S\in \ESL(2,\mathbb{F}_N)$ such that
\begin{equation}
F = S  \begin{pmatrix} 0 & - \Delta \\ 1 & t \end{pmatrix} S^{-1}
\end{equation}
\item If $t^2 - 4 \Delta =0$ and $\beta \neq 0$ then, for any $u \in \mathbb{F}^{*}_N$, there exists $S\in \ESL(2,\mathbb{F}_N)$ such that
\begin{equation}
F =  S \begin{pmatrix} 0 & \frac{ \Delta}{\beta u^2}  \\ -\beta u^2 & t \end{pmatrix} S^{-1}
\end{equation}
\item If $t^2 - 4 \Delta =0$ and $\gamma \neq 0$ then, for any $u \in \mathbb{F}^{*}_N$, there exists $S\in \ESL(2,\mathbb{F}_N)$ such that
\begin{equation}
F = S \begin{pmatrix} 0 & -\frac{ \Delta}{\gamma u^2}  \\ \gamma u^2 & t \end{pmatrix}S^{-1}
\end{equation}
\end{enumerate}
Moreover in every case the matrix $S$ can be chosen so that $\det S = \Delta$.
\end{lemma}
\begin{proof}
Let $k \in \mathbb{F}^{*}_N$ be arbitrary and look for a matrix
\begin{equation}
S = \begin{pmatrix} x & z \\ y & w \end{pmatrix}
\end{equation}
such that 
\begin{equation}
F = S \begin{pmatrix} 0 & -\frac{\Delta }{k} \\ k & t \end{pmatrix} S^{-1}
\end{equation}
It is easily seen that $S$ has this property if and only if
\begin{equation}
F \begin{pmatrix} x \\ y \end{pmatrix} = k \begin{pmatrix} z \\ w \end{pmatrix}
\end{equation}
which means $S$ must be of the form
\begin{equation}
S = \begin{pmatrix} x & \frac{\alpha x + \beta y}{k} \\ y & \frac{\gamma x + \delta y}{k} \end{pmatrix}
\end{equation}
Taking into account the requirement $\det S = \Delta$ the problem thus reduces to the problem of finding $x$, $y$ such that
\begin{equation}
\gamma x^2 + (\delta - \alpha) x y - \beta y^2 = k \Delta
\label{eq:xyCondA}
\end{equation}

\subsubsection*{Case 1:  $t^2 - 4 \Delta \neq 0$} Setting $k=1$ in Eq.~(\ref{eq:xyCondA}) the condition becomes
\begin{equation}
\gamma x^2 + (\delta - \alpha) x y - \beta y^2 = \Delta
\end{equation}
If $\gamma = 0$ this equation has the solution $x=(\Delta + \beta)/(\delta-\alpha)$, $y=1$ (note that the fact that $t^2 - 4 \Delta \neq 0$ means $\delta-\alpha \neq 0$).  If, on the other hand, $\gamma \neq 0$ we can rewrite the equation in the form
\begin{equation}
\left( x + \frac{(\delta -\alpha) y}{2 \gamma}\right)^2 = \frac{(t^2-4\Delta) y^2 + 4 \gamma \Delta}{4 \gamma^2}
\end{equation}
It follows from Lemma~\ref{lem:munueq} that there exist $r, s\in \mathbb{F}_N$ such that $(t^2-4\Delta) s^2 + 4 \gamma \Delta=r^2$.  The equation then has the solution $x = \left(r-s(\delta-\alpha)\right)/(2 \gamma)$, $y= s$.

\subsubsection*{Case 2:  $t^2 - 4 \Delta = 0$ and $\beta \neq 0$} Setting $k=-\beta u^2$ in Eq.~(\ref{eq:xyCondA}) the condition becomes, after rearranging, 
\begin{equation}
\left( y-\frac{(\delta -\alpha) x}{2 \beta}\right)^2 = \frac{u^2 t^2}{4}
\end{equation}
which has the solution $x=1$, $y = \left( u t \beta + (\delta - \alpha)\right)/(2 \beta)$.  

\subsubsection*{Case 3:  $t^2 - 4 \Delta = 0$ and $\gamma \neq 0$}  Proved in the same way as Case 2.

\end{proof}
We are now ready to deduce the conjugacy classes.  We write $F \sim G$ if $F$ is conjugate to $G$, and  use the symbol $[F]$ to denote the conjugacy class containing $F$.  Define
\begin{equation}
C_{\Delta,t} = \left[\begin{pmatrix} 0 & - \Delta \\ 1 & t\end{pmatrix}\right] 
\end{equation}
where $\Delta = \pm 1$ and $t$ is any element of $\mathbb{F}_N$.  We refer to these as the ``standard'' conjugacy classes.  We also need to consider the classes
\begin{equation}
\bar{C}_{t} = \left[\begin{pmatrix} 0 & - \frac{t^2}{4 \nu} \\ \nu & t\end{pmatrix}\right] 
\qquad \text{and} \qquad D_{ t}= \left[\begin{pmatrix} t/2 & 0 \\ 0 & t/2\end{pmatrix}\right]
\end{equation}
where $\nu$ is any fixed element of $\bar{Q}$ and $t$ is such that $t^2= \pm 4$.  We then have
\begin{theorem}
\label{thm:ESLconj1mod4}
Let $N = 1 \text{ (mod $4$)}$.  Let $i$ be one of the two elements of $\mathbb{F}_N$ with the property $i^2 = -1$.  Then the conjugacy classes of $\ESL(2,\mathbb{F}_N)$ comprise
\begin{enumerate}
\item The $2N$ standard classes $C_{\Delta, t}$  with $\Delta = \pm 1$ and $t\in \mathbb{F}_N$.
\item The $8$ additional classes  $\bar{C}_{\pm 2}$, $\bar{C}_{\pm 2 i}$, $D_{\pm 2}$, $D_{\pm 2 i}$.
\end{enumerate}
Let $F = \left(\begin{smallmatrix} \alpha & \beta \\ \gamma & \delta \end{smallmatrix}\right)$ and let $\Delta = \det F$, $t = \tr F$.  Then
\begin{enumerate}
\item If $t^2 - 4 \Delta\neq 0$
\begin{equation}
F \in C_{\Delta, t}
\end{equation}
\item If $t^2 - 4 \Delta = 0$
\begin{equation}
F \in \begin{cases} C_{\Delta, t} \qquad & \text{if $\beta$ or $\gamma \in Q$}\\
\bar{C}_{t} \qquad & \text{if $\beta$ or $\gamma \in \bar{Q}$}\\
D_{t} \qquad & \text{if $\beta = \gamma =0$}
\end{cases}
\end{equation}
(note that if $t^2 - 4\Delta = 0$ it cannot happen that $\beta \in Q$ and $\gamma \in \bar{Q}$, or that $\beta \in \bar{Q}$ and $\gamma \in Q$).
\end{enumerate}
\end{theorem}
\begin{remark}
So the determinant and trace completely fix the conjugacy class except when $t^2 - 4 \Delta =0$.
\end{remark}
\begin{proof}
To prove that the classes listed are disjoint note that $F\sim G$ implies $\det F = \det G$ and $\tr F = \tr G$, and that the classes $D_{ t}$ each consist of a single element.  It is therefore enough to show that 
\begin{equation}
C_{\Delta, t} \cap \bar{C}_t= \emptyset
\end{equation}
when $t^2 - 4 \Delta = 0$.  To see this suppose $F = \left( \begin{smallmatrix} \alpha & \beta \\ \gamma & \delta\end{smallmatrix}\right)\in C_{1,\pm 2}$.  Then there exists $S= \left( \begin{smallmatrix} x & z \\ y & w\end{smallmatrix}\right)\in \ESL(2,\mathbb{F}_N)$ such that
\begin{equation}
F = S \begin{pmatrix} 0 & -1 \\ 1 & \pm 2\end{pmatrix} S^{-1}
\end{equation}
Performing the algebra we find
\begin{equation}
\begin{pmatrix} \alpha & \beta \\ \gamma & \delta \end{pmatrix}
=
\begin{pmatrix} \Delta_S (z\mp x)(w \mp y) \pm 1 & - \Delta_S (z \mp x)^2 \\ \Delta_S (w \mp y)^2 & -\Delta_S (z\mp x)(w\mp y)\pm 1\end{pmatrix}
\end{equation}
where $\Delta_S = \det S$.  It follows that $\gamma \in Q \cup \{0\}$.  We conclude that $\left(\begin{smallmatrix} 0 & -1/\nu \\ \nu & \pm 2 \end{smallmatrix}\right) \notin C_{1,\pm2}$.  The proof that the classes $C_{-1,\pm 2 i}$ and $\bar{C}_{\pm 2 i}$ are disjoint is similar.

The remainder of the theorem is a straightforward consequence of Lemma~\ref{lem:simTrans}.

\end{proof}

If $N = 3 \text{ (mod $4$)}$ there are fewer conjugacy classes.  The reason is that in this case $-1 \in \bar{Q}$.  Consequently the classes $D_{\pm 2 i}$ and $\bar{C}_{\pm 2 i}$ do not exist.  Also $-\nu \in Q$, implying $-\nu = q^2$ for some $q \in \mathbb{F}^{*}_N$.  So the similarity transformation with $S = \left(\begin{smallmatrix} q & 0 \\ 0 & - 1/q\end{smallmatrix}\right)$ takes $\left(\begin{smallmatrix} 0 &  -1/\nu \\ \nu & \pm 2\end{smallmatrix}\right)$ to $\left(\begin{smallmatrix} 0 &  -1 \\ 1 & \pm 2\end{smallmatrix}\right)$, implying $\bar{C}_{\pm 2}= C_{1,\pm 2}$.  We thus have

\begin{theorem}
\label{thm:ESLconj3mod4}
Let $N = 3 \text{ (mod $4$)}$.   Then the conjugacy classes of $\ESL(2,\mathbb{F}_N)$ comprise
\begin{enumerate}
\item The $2N$ standard classes $C_{\Delta, t}$  with $\Delta = \pm 1$ and $t\in \mathbb{F}_N$.
\item The $2$ additional classes $D_{\pm 2}$.
\end{enumerate}
Let $F$ be any element of $\ESL(2,
\mathbb{F}_N)$ and let $\Delta = \det F$, $t = \tr F$.  Then
\begin{enumerate}
\item If $F\neq \pm I$
\begin{equation}
F \in C_{\Delta, t}
\end{equation}
\item If $F=\pm I$
\begin{equation}
F \in D_{\pm 2}
\end{equation}
\end{enumerate}
\end{theorem}
\begin{remark}
Here too the determinant and trace completely fix the conjugacy class except when $t^2 - 4\Delta = 0$.
\end{remark}
Finally, for the sake of completeness, let us note that Lemma~\ref{lem:simTrans} can also be used to deduce the conjugacy classes of $\SL(2,\mathbb{F}_N)$ (which are given by, for example, Gehles~\cite{12}).  Since $\Delta = 1$ for all elements of   $\SL(2,\mathbb{F}_N)$ we use the simplified notation (for the standard classes)
\begin{equation}
C_t = \left[ \begin{pmatrix} 0 & - 1\\ 1 & t\end{pmatrix} \right]
\end{equation}
and (for the additional classes)
\begin{equation}
\bar{C}_{\pm 2} = \left[ \begin{pmatrix} 0 & - \frac{1}{\nu} \\ \nu & 2\end{pmatrix} \right] \qquad \text{and} \qquad
D_{\pm 2} = \left[\begin{pmatrix}  \pm 1  & 0 \\ 0 &  \pm 1\end{pmatrix}\right]
\end{equation}
where $\nu$ is any fixed element of $\bar{Q}$.  We then have
\begin{theorem}
\label{thm:SLConjClass}
The conjugacy classes of $\SL(2,\mathbb{F}_N)$ comprise
\begin{enumerate}
\item The $N$ standard classes $C_t$, for arbitrary $t\in \mathbb{F}_N$.
\item The $4$ additional classes $D_{\pm 2}$, $\bar{C}_{\pm 2}$.
\end{enumerate}

Let $F=\left(\begin{smallmatrix} \alpha & \beta \\ \gamma & \delta \end{smallmatrix}\right)\in \SL(2,\mathbb{F}_N)$, and let $t=\tr F$.  Then
\begin{enumerate}
\item If $t\neq \pm 2$
\begin{equation}
F \in C_{t}
\end{equation}
\item If $t = \pm 2$
\begin{equation}
F \in \begin{cases} C_{t} \qquad & \text{if $-\beta$ or $\gamma \in Q$}\\
\bar{C}_{t} \qquad & \text{if $-\beta$ or $\gamma \in \bar{Q}$}\\
D_{t} \qquad & \text{if $\beta = \gamma =0$}
\end{cases}
\end{equation}
(note that if $t=\pm 2$ it cannot happen that $-\beta \in Q$ and $\gamma \in \bar{Q}$, or that $-\beta \in \bar{Q}$ and $\gamma \in Q$).
\end{enumerate}
\end{theorem}
\begin{remark}  So the trace completely determines the conjugacy class except when $t=\pm 2$.  Note that for the group $\SL(2,\mathbb{F}_N)$ there is no distinction between the cases $N = 1 \text{ (mod $4$)}$ and $N = 3 \text{ (mod $4$)}$.
\end{remark}
\begin{proof}
Similar to the proof of Theorem~\ref{thm:ESLconj1mod4}.
\end{proof}


\begin{thebibliography}{99}
\bibitem{1} E. P. Wigner, Phys. Rev. \textbf{40} 749 (1932) .  For
reviews see: M. Hillery, R. F. O'Connell, M. O. Scully and E. P. Wigner, 
{\it Phys. Repts.} \textbf{106} 121, (1984); Y. S. Kim and M. E. Noz, {\it Phase-Space 
Picture of Quantum Mechanics} (World Scientific, Singapore, 1991); W. P. Schleich, 
{\it Quantum Optics in Phase Space} (Wiley-VCH, Weinheim, 2001)

\bibitem{2} F. A. Buot, Phys. Rev. B {\bf 10}, 3700 (1974);  R. Jagannathan, {\it Studies in Generalized Clifford Algebras, Generalized Clifford Groups and their Physical Applications} Ph D thesis (University of Madras) (1976);  N. Mukunda, Am. J. Phys. {\bf 47}, 182 (1979); J. H. Hannay and M. V. Berry ,Physica D {\bf 1} 26 (1980); L. Cohen and M. Scully, Found. Phys. {\bf 16}, 295 (1986); R. P. Feynman R P  in {\it Quantum 
Implications. Essays in Honour of David Bohm} Eds. B. Hiley and D. Peat, 
(Routledge, London, 1987); W. K. Wootters, Ann. Phys. (N.Y.) \textbf{176} 1 (1987).

\bibitem{3} O. Cohendet, P. Combe, M. Siugue and M. Sirugue-Collin M,  J. Phys. A {\bf 21} 2875 (1988) ; D. Galetti and A. F. R. de Toledo Piza, Physica  A {\bf 149} 267(1988) ; 
J. A. Vacarro  and D. T. Pegg,  Phys. Rev. A {\bf 41}, 5156 (1990);
P. Kasperkovitz  and M. Peev, Ann. Phys. (N. Y.) {\bf 230} 21 (1994); 
A. Bouzouina  and S. Bi\`evre,  Comm Math. Phys. {\bf 178}, 83 (1996) ;
U. Leonhardt, Phys. Rev. A{\bf 53} 2998 (1996),  and  Phys. Rev. Lett. {\bf 76}, 4293 (1996); A. M. Rivas  and A. M. Ozorio de Almeida, Ann. Phys. (N. Y.) {\bf 276}, 123 (1999); 
M.Ruzzi and D. Galetti, J. Phys. A {\bf 33}, 1065 (1999); M. Horibe, A. Takami, T. Hashimoto and A. Hayashi {\it  Phys. Rev. A} {\bf 65} 032105 (2002).  

\bibitem{4} N. Mukunda, Arvind, S. Chaturvedi and R. Simon,  
 J. Math. Phys. \textbf{45}, 114 (2004); S. Chaturvedi, E. Ercolessi, G. Marmo, G. Morandi, N. Mukunda and R. Simon, Pramana,J. Phys. {\bf 65}, 981 (2005); A. Vourdas, Rep. Prog. Phys. {\bf 67}, 267(2004); D. Gross, J. Math. Phys. {\bf 47}, 122107 (2006).  

\bibitem{5} K. S. Gibbons, M. J. Hoffman and W. K. Wootters, 
{\it Phys. Rev. A} {\bf 70} 062101 (2004) ; W. K. Wootters, IBM J. of Research and Development {\bf 48}, 99 (2004); W. K. Wootters, Foundations of Physics {\bf 36}, 112 (2006).

% \bibitem{6} M. Nielsen and I. Chuang, {\it Quantum Computation and Quantum Information}
%(Cambridge University Press, New York, 2000); I. Bengtsson and K. $\dot{{\rm Z}}$yczkowski, 
%{\it Geometry of Quantum States : An Introduction to Quantum Entanglement}
%(Cambridge University Press, New York, 2006). 

\bibitem{7} P. Bianucci, C. Miquel, J. P. Paz and M. Saraceno, 
Phys. Lett. A {\bf 297}, 353 (2002) ; 
R. Asplund and G. Bj\"ork, Phys. Rev. A {\bf 64}, 012106 (2001);
C. Miquel, J. P. Paz, M. Saraceno, E. Knill, 
R. Laflamme and C. Negrevergne, Nature (London) {\bf 418}, 59 (2002) ; C. Miquel, Paz J P and M. Saraceno, Phys. Rev. A {\bf 65}, 062309 (2002);  
J. P. Paz, Phys. Rev. A {\bf 65}, 062311 (2002);   
J. P. Paz, A. J. Roncaglia and M. Saraceno, Phys. Rev. A{\bf 72}, 012309(2004);



\bibitem{8} J. Schwinger, {\it Proc. Nat. Acad. Sci. USA}  {\bf 46} 570 (1960); I. D. Ivanovic,  J. Phys. A {\bf 14}, 3241(1981); W. K. Wootters 
 and  B. D. Fields, Ann. Phys. (N.Y.) {\bf 191}, 363 (1989); A. R. Calderbank 
, P. J. Cameron, W. M. Kantor and J. J. Seidel, Proc. London. Math. 
Soc. {\bf  75} 436(1997); S. Bandyopadhyay, P. O. Boykin, V. Roychowdhury and F. Vatan, 
 Algorirhmica, {\bf 34} 512 (2002) ; J. Lawrence, C. Brukner and A. Zeilinger, 
 Phys. Rev. A {\bf 65} 032320 (2002); S. Chaturvedi,  
Phys. Rev. A {\bf 65} 044301 (2002) ; A. O. Pittenger and M. H. Rubin, 
 Linear Alg. Appl. {\bf 390} 255 (2004) ; A. O. Pittenger and  M. H. Rubin, 
 J. Phys. A {\bf 38} 6005 (2005); A. Klappenecker and M. R\"otteler, Lecture Notes in Computer Science  {\bf 2948} 137 (2004); K. R. Parthasarathy, Anal. Quantum Probab. Relat. Top. {\bf 7}, 607 (2004).
 
\bibitem{9} M. Saniga, M. Planat and H. Rosu, J. Opt.  Quantum 
Semiclass. B{\bf 6}, L19 (2004) ; P. Wocjan and T. Beth, Quantum Information 
and Computation {\bf 5}, 93 (2005); A. Hayashi A, M. Horibe M and Hashimoto T Phys. Rev. A 71, 052331 (2005)  ; H. Barnum, eprint, quant-ph/0205155 (2002);  A. Klappenecker and M. R\"otteler Proc. 2005 IEEE International Symposium on Information Theory, Adelaide, Australia, pp. 1740-1744, 2005; 
\bibitem{10} D. M. Appleby, J. Math. Phys. {\bf 46}, 052107 (2005).
\bibitem{11}G. Zauner  {\it Quantumdesigns: Grundz\"uge einer nichtkommutativen Designtheorie} Ph.D. thesis (Universit\"at Wien) (1999);  
J. M. Renes, R. Blume-Kohout, A. J. Scott, and C. M. Caves, J. Math. Phys. {\bf 45}, 2171 (2004); 
%I. Bengtsson, AIP Conf. Proceedings {\bf 750}, 63 (2005); I. Bengtsson, and \AA. 
%Ericsson, Open Systems and Information Dynamics, {\bf 12}, 107 (2005); 
M. Grassl, %"On SIC-POVMs and MUBs in dimension 6", 
eprint, quant-ph/0406175. 

\bibitem{AF} D. M. Appleby, H. B. Dang, and C. A. Fuchs, eprint, arXiv:0707.2071.
\bibitem{12} See J. E. Humphreys, Amer. Math. Month. {\bf 82}, 21 (1975), or 
K. E. Gehles, MSc Thesis, Univ. of St Andrews 2002.  Also see S.T.~Flammia, \emph{J. Phys. A}, \textbf{39}, 13483 (2006).
\bibitem{Hardy} G.H.~Hardy and E.M.~Wright, \emph{An Introduction to the Theory of Numbers}, fifth edition (Clarendon Press, Oxford, 1979).
\bibitem{Lidl} R.~Lidl and H.~Niederreiter, \emph{Finite Fields}, Encyclopedia of Mathematics and its Applications 20, second edition (Cambridge University Press, Cambridge, 1997).





\end{thebibliography}
\end{document}